\documentclass[sigconf, nonacm]{aamas}

\usepackage{balance} % for balancing columns on the final page

\usepackage{bm}
\usepackage{algorithm}
\usepackage{algpseudocode}
\usepackage{array}
\usepackage{amsthm}
\newtheorem{lemma}{Lemma}
\DeclareMathOperator*{\argmax}{arg\,max}

\theoremstyle{definition}
\newtheorem{definition}{Definition}

\thanks{*Equal contribution}
\thanks{$^{1}$Ayhan Alp Aydeniz and Kagan Tumer are with Collaborative Robotics and Intelligent Systems Institute, Oregon State University,
         Corvallis, Oregon, USA
         {\tt\small \{aydeniza, kagan.tumer\}@oregonstate.edu}}%
\thanks{$^{2}$Enrico Marchesini is with Laboratory for Information \& Decision Systems, Massachusetts Institute of Technology, Cambridge, Massachusetts, USA
         {\tt\small e.marchesini@mit.edu}}%
\thanks{$^{3}$Robert Loftin is with the Department of Computer Science, University of Sheffield, Sheffield, United Kingdom
         {\tt\small r.loftin@sheffield.ac.uk}}
\thanks{$^{4}$Christopher Amato is with Khoury College of Computer Sciences, Northeastern University, Boston, Massachusetts, USA
         {\tt\small c.amato@northeastern.edu}}

\title[]{Safe Multiagent Coordination via Entropic Exploration}

\author{Ayhan Alp Aydeniz$^{1*}$, Enrico Marchesini$^{2*}$, Robert Loftin$^{3}$, Christopher Amato$^{4}$, and Kagan Tumer$^{1}$} 

\begin{abstract}
Many real-world multiagent learning problems involve safety concerns. In these setups, typical safe reinforcement learning algorithms constrain agents' behavior, limiting exploration---a crucial component for discovering effective cooperative multiagent behaviors. Moreover, the multiagent literature typically models individual constraints for each agent and has yet to investigate the benefits of using joint (\textit{team}) constraints. In this work, we analyze these team constraints from a theoretical and practical perspective and propose \textit{entropic exploration} for constrained multiagent reinforcement learning (E2C) to address the exploration issue. E2C leverages observation entropy maximization to incentivize exploration and facilitate learning safe and effective cooperative behaviors. Experiments across increasingly complex domains show that E2C agents match or surpass common unconstrained and constrained baselines in task performance while reducing unsafe behaviors by up to $50\%$.

\end{abstract}


\newcommand{\BibTeX}{\rm B\kern-.05em{\sc i\kern-.025em b}\kern-.08em\TeX}

\begin{document}

%%% The following commands remove the headers in your paper. For final 
%%% papers, these will be inserted during the pagination process.

\pagestyle{fancy}
\fancyhead{}

%%% The next command prints the information defined in the preamble.

\maketitle 

%%%%%%%%%%%%%%%%%%%%%%%%%%%%%%%%%%%%%%%%%%%%%%%%%%%%%%%%%%%%%%%%%%%%%%%%

\section{Introduction}
\label{sec:introduction}

Training agents to operate in real-world scenarios requires adhering to safety specifications. This is particularly challenging in multiagent environments (\textit{e.g.,} search-and-rescue missions~\cite{mr_searchrescue} and remote exploration tasks~\cite{hu2020voronoi}), where agents must discover highly coordinated behaviors while satisfying the safety requirements. However, learning these behaviors is difficult as it requires optimizing for (at least) two objectives---a team objective aiming to solve the task and the safety objective(s). Hence, extending single-agent safe reinforcement learning (RL) algorithms to multiagent settings often fails to cope with the dynamics and challenges of cooperative systems.

In particular, prior safe RL work has explored constrained algorithms using trust region-based methods to match the safety specifications~\cite{srinivasan2020learning, thananjeyan2021recovery, thananjeyan2020safety, achiam2017constrained, constrained_behaviors}. These constrained RL methods aim to maximize task performance---modeled as reward signal(s)---while adhering to safety specifications in the form of constraints. Despite being effective in single-agent settings, their extensions to multiagent RL (MARL) model the problem by defining separate (\textit{individual}) constraints for each agent~\cite{constrained_mappo, constrained_mappo2}. However, we argue that safety in cooperative MARL is inherently a team-level concept (\textit{e.g.,} a collision between two agents can cause the entire team to fail). Additionally, introducing constraints inherently limits exploration, which is critical for avoiding local optima~\cite{diversity_exploration}. Without effective exploration, constrained algorithms tend to prioritize safety at the expense of learning behaviors that achieve strong task performance. This detrimental trade-off further exacerbates in multiagent cooperative systems, where exploration is crucial for discovering joint behaviors that are necessary to solve a task~\cite{mappo, maddpg, liir, qmix}. To improve discovery of such behaviors, constrained MARL often considers policy entropy in the optimization process~\cite{ha2020learning, yang2023reinforcement}. However, we note that this technique increases the ``randomness" of policy decisions and could further hinder constraint satisfaction and the team's performance in complex tasks.
\begin{figure}[t]
    \centering
    \includegraphics[width=0.8\linewidth]{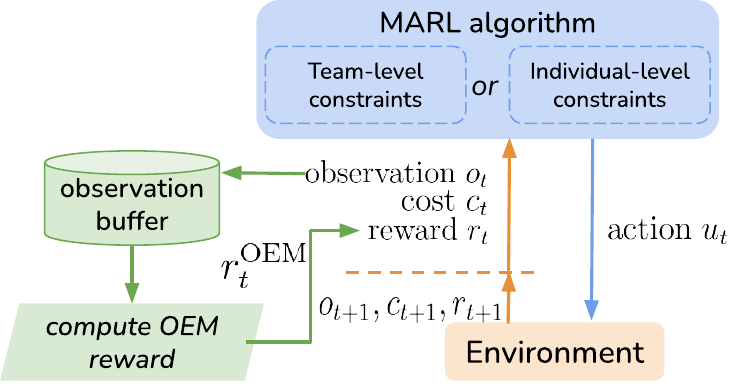}
    \caption{Explanatory overview of an E2C algorithm.}
    \label{fig:E2C_overview}
    \vspace{-5pt}
\end{figure}

In this paper, we address these issues by first analyzing the impact of defining team constraints from a theoretical perspective. Then, we introduce entropic exploration for constrained MARL (E2C) to enhance exploration and learn safe behaviors with good team performance for cooperative agents (Figure \ref{fig:E2C_overview}). E2C employs observation entropy maximization (OEM)~\cite{seo2021state, liu2021behavior, bellemare2016unifying} to balance the team objective of agents and the safety objective in the system without employing policy entropy. Specifically, E2C leverages a \textit{count-based} and \textit{k-nearest neighbor (knn)} approximations to estimate the observation entropy and reward the agents. Finally, we apply E2C to a constrained MARL algorithm with both individual and team constraints to analyze their practical impact. The contributions of this work are to:

\begin{itemize} 
    \item Introduce team constraints for cooperative agents, providing a lower bound on policy improvement and showing their practical impact on performance. 
    \item Propose a novel constrained MARL algorithm leveraging entropy maximization to balance task performance and safety constraints effectively.
\end{itemize}

Our experiments show the efficacy of E2C in addressing safety and cooperation in a variety of multiagent setups. We evaluate our method across six well-known MARL domains, including variations of the \textit{cooperative multi-rover exploration} domain~\cite{agogino2004unifying}, three coordination environments from the \textit{multiagent particle suite}~\cite{maddpg}, and two multiagent locomotion tasks from the \textit{safe MaMuJoCo} problems~\cite{safe_mamujoco}.
Notably, E2C algorithms achieve superior or comparable task performance to baseline unconstrained MARL algorithms, while satisfying constraints. Our method also successfully satisfies constraints and learns high-payoff behaviors in complex coordination tasks where previous constrained MARL baselines fail. 
\section{Preliminaries and Related Work}
\label{sec:preliminaries}

Cooperative multiagent tasks can be modeled as a \textit{decentralized Markov decision processes} (DecMDPs)~\cite{decpomdp}. We represent a Dec-MDP with a tuple $\langle \mathcal{N}, \mathcal{S}, \mathcal{U}, T, r, O, \gamma \rangle$, where $\mathcal{N}$ is a finite set of agents; $\mathcal{S}$ is the set of states of the environment; $\mathcal{U} = \{U^i\}_{i \in \mathcal{N}}$ is the set of all possible joint actions. At a time step, $t$, each agent $i$ selects an action, $u_t^i$, forming joint action $\bm{u} = \{u_t^i\}_{i \in \mathcal{N}}$, which transitions the environment from state $s_t$ to $s_{t+1}$ via $T(s_t, \bm{u}_t, s_{t+1}) = P(s_{t+1}|s_t, \bm{u}_t)$ and yields a joint reward $r(s_t, \bm{u}_t)$. In a DecMDP, each agent $i$ receives an observation $o_t^i$ according to an observation transition function $O(\bm{o}_t, s_t, \bm{u}_t) = P(\bm{o}_t | s_t, \bm{u}_t)$. This function defines the probability distribution over the joint observations $\bm{o}_t = \{o_t^1, \dots, o_t^n \}$ based on the state and joint action. In a DecMDP, the state at each time step can be uniquely determined by the joint observation. The objective of a cooperative team of agents is to learn a joint policy $\pi(\bm{u}_t|s_t)$ that maximizes the expected discounted return defined as: 
\begin{equation}
J_r(\pi) := \mathbb{E}_{\pi} \left[\sum_{t=0}^{\infty}\gamma^t r(s_t, \bm{u}_t) \right],
\end{equation}
\noindent where $\gamma \in [0, 1)$ is a discount factor.

\subsection{Multiagent Reinforcement Learning}
\label{sec:background_marl}
Motivated by the promises of collaborative multiagent systems, MARL has received significant research attention.  A popular approach is the \textit{centralized training with decentralized execution (CTDE)} paradigm, which centralizes information during training while maintaining decentralized execution \citep{mappo, coma, qplex}.
In discrete action spaces, value factorization techniques have been used to estimate a joint value function as a global optimization signal~\cite{qmix, qplex, gdq} in a CTDE fashion. However, these algorithms are unsuitable for multiagent systems with continuous control.

To address this limitation, extensions of well-known single-agent algorithms to MARL (\textit{e.g.}, proximal policy optimization (PPO)~\cite{ppo} to multiagent PPO (MAPPO)~\cite{mappo}) have shown promising performance in cooperative games with both discrete and continuous action spaces, by estimating centralized value functions~\cite{mappo, maddpg, coma}. Crucially, the centralized component in these approaches is required only during training, ensuring a principled CTDE method. 

Given MAPPO's strong performance in decentralized multiagent problems, we build E2C on top of it. Next, we discuss how to incorporate safety specifications within this learning paradigm.

\subsubsection{Constrained Reinforcement Learning}
\label{sec:constrained_rl}
Constrained RL has successfully modeled safety specifications for single-agent tasks~\cite{safetygym, constrained_behaviors, spaan_cem}. These constrained methods typically extend trust region-based algorithms by incorporating additional sets of cost functions, denoted as $\mathcal{C} := \{c_j^i\}_{j \in m^i}^{i \in \mathcal{N}}$, where each agent $i$ has $m^i$ cost functions. Each cost function is of the form $c_j^i: \mathcal{S} \times U^i \rightarrow \{0, 1\}$ with corresponding hard-coded cost-limiting values (\textit{i.e.}, thresholds) $\mathbf{l} := \{l_j^i\}_{j \in m^i}^{i \in \mathcal{N}}$. After performing the joint action in the environment, each agent $i$ also receives costs $c_j^i(s_t, u_t^i)~\forall j = 1, \dots, m^i.$ In addition to maximizing the expected discounted return, the agent aims to satisfy its safety constraints, defined as:
\begin{equation}
J_j^i(\pi) := \mathbb{E}_{\pi} \left[\sum_{t=0}^{\infty}\gamma^t c_j^i(s_t, u_t^i) \right] \leq l_j^i,~\forall j = 1, \dots, m^i.
\end{equation}
Constraint-satisfying (feasible) policies $\Pi_\mathcal{C}$ (where $\Pi$ are stationary policies), and optimal policies $\pi^*$ are thus defined as:
\begin{equation}
\begin{split}
    &\Pi_\mathcal{C} := \{\pi \in \Pi : J_j^i(\pi) \leq l_j^i,~\forall i \in \mathcal{N}, j = 1, \dots, m^i\}, \\
    &\pi^*= \argmax_{\pi \in \Pi_\mathcal{C}} J_r(\pi).
\label{eq:id_optimum}
\end{split}
\end{equation} 
However, enforcing strict thresholds leads to a significant performance trade-off due to the detrimental effects of constraints on exploration~\cite{AAAI, safetygym, spaan_cem}. Such drawbacks are further exacerbated in multiagent settings that we discuss in the following.

\subsubsection{Constrained Multiagent Reinforcement Learning}
In cooperative MARL, good exploration is pivotal to learning a joint policy that can successfully solve a cooperative task. In particular, introducing constraints in MARL typically leads to having three competing objectives: agent-specific behaviors (\textit{e.g.,} learning how to navigate in an environment), joint task performance (\textit{e.g.,} cooperating to rescue a target), and constraints (\textit{e.g.,} avoiding collisions). 

% (ii) disregarding the negative impact of constraints and policy entropy to learn cooperative behaviors.

Despite its importance, constrained MARL has received marginal research attention~\cite{constrained_mappo, constrained_mappo2, constrained_mappo3} and there is much room to develop novel safe MARL algorithms. For example, the works~\cite{constrained_mappo2, constrained_mappo3} propose approaches voted to improve cost-value estimation and credit assignment. However, these methods extend the single-agent case by: (i) using individual constraints for each agent, (ii) disregarding the negative impact of constraints and policy entropy used during optimization by the algorithms to learn cooperative behaviors. We thus analyze the benefit of incorporating safety specifications at a team level from a theoretical and practical perspective. On the theory side, we extend the work of~\citet{constrained_mappo}, deriving cost improvement bounds for trust-region-based methods (on which MAPPO builds~\cite{cmarl_bounds}) using team constraints. On the practical side, we conduct an extensive evaluation of constrained MAPPO employing individual and team constraints. Moreover, to tackle the limited exploration of constrained algorithms and the potential issues of employing policy entropy, by integrating observation entropy maximization for which we provide a brief overview in the following.

%\subsection{Entropy Maximization in Reinforcement Learning}
\subsection{Entropy Maximization}
\label{sec:ent_max_in_RL}
A popular technique to improve exploration in RL algorithms is to generate diversity that enables the policy search to cover larger areas in the policy space. Many RL algorithms achieves this search by maximizing the policy entropy~\cite{eysenbach2018diversity, haarnoja2018soft, ziebart2008maximum, ppo} within their optimization processes or as their objective. This entropy maximization naturally increases stochasticity while agents take actions and has been used within constrained RL frameworks~\cite{ha2020learning, yang2023reinforcement}. In this work, we show that maximizing policy entropy is not a good practice when cooperative agents are under strict safety requirements. 

Conversely, observation (or state) entropy maximization (OEM) is used to incentivize visiting new states by rewarding agents based on the novelty of their observations. To this end, designing a count-based reward has been a popular approach for mapping the agent's decisions to the entropy measured over the observed states~\cite{bellemare2016unifying, badia2020never, tang2017exploration, aydeniz2023novelty}. In a similar direction, \textit{k-nearest neighbor} (\textit{knn}) estimates of entropy~\cite{beirlant1997nonparametric, singh2003nearest} have been used to have a uniform distribution of agents' observations~\cite{seo2021state, liu2021behavior, aydeniz2023entropy}. The common practice of maximizing observation entropy is to design rewards promoting the visitations to unique states. Crucially, these entropy-based rewards impact the magnitude of gradients and not their direction, potentially addressing the issues of using policy entropy in constrained setups. However, to our knowledge, OEM has yet to be explored as a practical way to address the exploration issues of constrained MARL setups.

\section{Trust Region Bounds for Team Constraints}
\label{sec:team_theory}

In this section, we extend the cost improvement bounds derived by the works~\cite{cmarl_bounds, constrained_mappo} for trust region MARL with individual constraints to the team settings. In particular, we note that the work \cite{constrained_mappo} considers a Markov game, assuming fully cooperative agents with a joint reward, and \citet{cmarl_bounds} relies on a Dec-MDP formalization assuming local observations capture ``sufficient" information about the state.\footnote{When their state assumption does not hold, authors assume that agents using recurrent networks as decentralized policies can overcome partial observability.} Hence, we follow the same assumptions of such previous works and extend their stateful lower bound on the cost improvement to team constraints.

When cooperative agents use joint (team) constraints, we define a set of cost functions $\mathcal{C} := \{c_j\}_{j \in m}$ (the team has $m$ cost functions). These functions take the form $c_j: \mathcal{S} \times \mathcal{U} \rightarrow \{0, 1\}$ with cost-limiting values $\bm{l} := \{l_j\}_{j \in m}.$ After performing the joint action in the environment, the agents receive joint costs $c_j(s_t, \bm{u}_t)~\forall j = 1, \dots, m.$ On top of maximizing the expected discounted return, the agents now also try to satisfy a joint constrained objective:
\begin{equation}
    J_j(\pi) := \mathbb{E}_{\pi} \left[\sum_{t=0}^{\infty}\gamma^t c_j(s_t, \bm{u}_t) \right] \leq l_j,~\forall j = 1, \dots, m,
\label{eq:team_optimum}
\end{equation}
\noindent for which optimal policies are defined similarly to Equation \ref{eq:id_optimum} by replacing the cost objectives in the space of feasible policies $\Pi_\mathcal{C}$.

To derive the cost improvement bound for team constraints, we define the corresponding joint cost value functions. For the $j^{\text{th}}$ cost function, we define the $j^{\text{th}}$ (stateful) value functions as follows:

\begin{equation}
\begin{split}
    &Q_j^\pi(s, \bm{u}) := \mathbb{E}_{\pi}\left[\sum_{t=0}^{\infty}\gamma^t c_j(s_t, \bm{u}_t) \vert s_0 = s, \bm{u}_0 = \bm{u}\right], \\
    &V_j^\pi(s) := \mathbb{E}_{\bm{u}\sim\pi}\left[Q_{j, \pi}(s, \bm{u}) \right], \\
    &A_j^\pi(s, \bm{u}) = Q_{j, \pi}(s, \bm{u}) - V_{j, \pi}(s).
\end{split}
\end{equation}

In trust region-based methods, Equation \ref{eq:team_optimum} is difficult to optimize directly when considering a joint policy $\pi$ and some other policy $\bar\pi^i$ of agent $i$. Hence, we define the surrogate objective for team constraints following the individual constraint case of \citet{constrained_mappo}.

\begin{definition}
    Let $\pi$ be a joint policy, and $\bar{\pi}^i$ be some other policy of agent $i$. Then, for any of the joint costs of index $j = 1, \dots, m$, we define the surrogate cost objective as follows:
    \[
        L_j^\pi(\bar\pi^i) = \mathbb{E}_{\bm{u}^{-i}\sim \pi^{-i}, u^i \sim \bar{\pi}^i}\left[A_j^\pi(s, \bm{u}) \right],
    \]
\end{definition}
\noindent where $\pi^{-i}$ indicates the policy of all the agents except $i$.

Finally, we extend Lemma 4.3 of \citet{constrained_mappo} to the case of team constraints, deriving a lower bound on how the expected joint costs change when the agents update their policies.

\begin{lemma}
\label{lemma:team_bound}
    Let $\pi$ and $\bar{\pi}$ be joint policies. Let $i\in\mathcal{N}$ be an agent, and $j = 1, \dots, m$ be one of the joint cost indexes. The following inequality holds:
    \begin{equation*}
    \begin{split}  
        &J_j(\bar{\pi}) \leq J_j(\pi) + L_j^\pi(\bar{\pi}^i)+\nu_j\sum_{h = 1}^{|\mathcal{N}|} D_{KL}^{max}(\pi^h, \bar{\pi}^h),\\
        &\text{where}~ \nu_j = \frac{4\gamma\max_{s, \bm{u}}\vert A_j^\pi(s, \bm{u})|}{(1-\gamma)^2}.
    \end{split}
    \end{equation*}
\end{lemma}
\noindent Proof of Lemma \ref{lemma:team_bound} is presented in Appendix \ref{app:proof}. In practice, trust region algorithms ensuring Lemma \ref{lemma:team_bound} (or its equivalent version for the individual constraints proposed by \citet{constrained_mappo}) are replaced by approximations relying on neural networks and tractable clipping operators that can scale to high-dimensional state and action spaces \cite{trpo, ppo, constrained_mappo}, on top of which we build E2C in the following section.

\section{Entropic Exploration for Constrained MARL}
After introducing the team constraints problem, we present our E2C method that addresses the limited exploration arising from constrained MARL settings (with both individual and team constraints). In particular, we remark that constrained MARL algorithms typically result in a detrimental trade-off between safety and task performance. Exploration plays a crucial role for discovering joint behaviors and balance this trade-off. Before introducing E2C, this section formalizes the constrained MARL problem for both individual and team constraint cases.

The typical goal of constrained MARL is to maximize a reward function while satisfying constraints modeling safe behaviors. Recalling the cost and reward-based objective notations of Sections \ref{sec:background_marl} and~\ref{sec:team_theory},  constrained problems can be defined as maximizing the joint reward objective---$\max_\pi J_r(\pi)$---while satisfying constraints---$J_j^i(\pi) \leq l_j^i~\forall i \in \mathcal{N},~j = 1, \dots, m^i$ for individual constraints, or $J_j(\pi) \leq l_j~\forall j = 1, \dots, m$ for team ones.

In the safe RL literature, the Lagrangian method~\citep{Nocedal} is commonly used to transform the constrained problem into an equivalent unconstrained one $\forall i \in \mathcal{N}$, using a dual variable as follows:
\begin{equation}
\begin{split}
&\mathcal{L}^{\pi}(\bm{\lambda}) = J_r^{\pi} - \mathcal{L}^{\pi}_\mathcal{C}(\bm{\lambda}), \\
&\mathcal{L}^{\pi}_\mathcal{C}(\bm{\lambda}) = 
\begin{cases}
    \lambda_j^i \big(J_j^i(\pi) - l_j^i\big)~\forall j = 1, \dots, m^i & \text{individual}\\
    \lambda_j \big(J_j(\pi) - l_j)~\forall j = 1, \dots, m & \text{team}\\
\end{cases},
\label{eq:e2c_lagrangian}
\end{split}
\end{equation}
\noindent where $\bm{\lambda}$ are the so-called \textit{Lagrangian multipliers} and act as a penalty in the optimization objective of each agent. The goal is thus to solve the resulting \textit{max min} problem: 
\begin{equation}
\max_\pi \min_{\bm{\lambda} \geq 0} \mathcal{L}^{\pi}(\bm{\lambda}).
\label{eq:maxmin}
\end{equation}
A typical solution to Equation \ref{eq:maxmin} is to iteratively take gradient ascent steps in $\pi$ and descent in $\bm{\lambda}$. We first update the multipliers following $\nabla_{\bm{\lambda}} \mathcal{L}^{\pi}(\bm{\lambda})$, noting the multipliers must be $\geq 0$ because they act as a penalty when the constraints are not satisfied (\textit{i.e.}, $\bm{\lambda}$ increases) while decreasing to $0$ (\textit{i.e.}, removing any penalty) when the constraint objectives are satisfied. Then, we maximize the policy's parameters following $\nabla_{\pi} \mathcal{L}^{\pi}\big(\bm{\lambda})$. Broadly speaking, Lagrangian-based algorithms focus on satisfying the imposed constraints using a penalty growing unbounded when constraints are violated. Once the constraints are satisfied, the multipliers scale down (to zero) and allow the gradient to follow the direction that maximizes the main reward objective. Such formalization can be used to optimize arbitrary policy gradient objectives.

\subsection{E2C}
\label{sec:e2c}
As discussed in Section \ref{sec:preliminaries}, we build E2C on top of the Lagrangian MAPPO---a strong baseline across a variety of scenarios~\citep{mappo, constrained_mappo}.\footnote{Due to the nature of the Lagrangian method, and the following discussion on OEM (via exploration-driven rewards), our contributions are orthogonal to the chosen policy-gradient MARL algorithm and can thus be integrated with other approaches.} The resultant E2C-MAPPO algorithms address the challenges of using constraints in multiagent systems by using \textit{entropy enhanced agents}. In this section, we start by deriving the constrained MAPPO algorithm for individual and team constraints and then present the entropic exploration method based on OEM.

Following the MAPPO baseline, we learn a centralized advantage estimator $A_\phi(s, \bm{u})$ parametrized by $\phi$, while each agent $i \in \mathcal{N}$ learns a policy $\pi_{\theta_i}$ parametrized by $\theta_i$. Policies' parameters are updated using the following clipped objective:
\begin{equation}
\begin{split}
\max_{\theta_i}\min_{\bm{\lambda}}~\mathbb{E}_{\pi_{\theta_i}}\Big[&\min\Big( q(\theta_i, \theta_i') A_\phi(s, \bm{u}), \\
&\hspace{23pt}\text{clip}\left(q(\theta_i, \theta_i'), 1-\epsilon,1+\epsilon\right)A_\phi(s, \bm{u}) \Big)\\
&+ q(\theta_i, \theta_i') \mathcal{L}_\mathcal{C}^{\pi_{\theta_i}}(\bm{\lambda})\Big],
\label{eq:constrained_mappo}
\end{split}
\end{equation}
where the centralized advantage measures the overall effect of selecting a joint action, and $q(\theta_i, \theta_i') = \frac{\pi_{\theta_i}(u_i|h_i)}{\pi_{\theta_i'}(u_i|h_i)}$.
In more detail, $\mathcal{L}_\mathcal{C}^{\pi_{\theta_i}}$ depends on the nature of constraints (\textit{i.e.,} individual or team as in Equation \ref{eq:e2c_lagrangian}). For example, consider the simplified case with one individual constraint for each agent and the corresponding one team constraint case. In the case of an individual constraint $c_1^i$ (with threshold $l_1^i$), each agent $i$ learns a separate multiplier $\lambda_1^i$ and a cost-advantage estimator based on local information. Conversely, when considering a team constraint $c_1$ with threshold $l_1$, we learn a single team multiplier $\lambda_1$ and a joint cost-advantage estimator. Both approaches have pros and cons that follow the benefits and drawbacks of using decentralized (\textit{i.e.,} with local information) and centralized (\textit{i.e.,} with joint information) estimators in MARL. Specifically, individual constraints scale better as the size of local information used by cost-advantage estimators does not depend on the number of agents. However, using local information can hinder performance \citep{iql}. In contrast, the centralized case does not scale well in the number of agents due to the cardinality of the joint observation and action spaces but leveraging joint information typically improves value estimates and performance \citep{maddpg}.

\subsubsection{Entropic Exploration} 

We use observation entropy maximization to design an exploration-driven reward that incentivizes agents' exploration. For OEM, we employ \textit{quantization} \citep{aydeniz2023novelty} to cluster similar observation vectors together for the experiments in multi-rover exploration \citep{agogino2004unifying} due to the low cardinality of the observations. In the more complex state spaces of particle environments and safe MaMuJoCo tasks \cite{penv_original, safe_mamujoco}, we use \textit{knn} estimate \citep{singh2003nearest} of entropy to deal with higher dimensional observations. The resultant OEM reward is presented in Algorithm \ref{alg:entropyShape}. Specifically, each agent receives a reward bonus based on the novelty of its current observation. OEM thus aims to improve the search in the observation space, which is key to learning good joint policies in multiagent systems \citep{aydeniz2023novelty}. To emphasize observations that might have a contribution to the overall team task more prominent (for an efficient search), we incorporate a value, $\beta(o)$, as described in~\cite{aydeniz2023novelty} in multi-rover domain. In safe MaMuJoCo tasks, we incorporate the OEM reward into the extrinsic task reward, $r_{ex}(o)$, via a hard-coded weight (0.3), $\psi$, as applied in~\cite{seo2021state}. In particle environments, we observe that using pure OEM rewards is sufficient to unearth safe and cooperative behaviors.

\begin{algorithm}[t]
\caption{Observation Entropy Maximizing reward for agent $i$}
\label{alg:entropyShape}
\begin{algorithmic}[1]
\State \textbf{Define} flags $count\_based$, $knn\_approximation$, and $mix\_with\_extrinsic\_reward$
% \State Initialize state $s^{0}$, history $h^{0}_i = \{ o^{0}_i\}$
\State For a sampled observation  $o_{t+1}^i$ by action $u_{t}^i$ and observation buffer $b^i$ storing observations from a particular episode
    % \State Retrieve action $a^{t}_i$, state $s^{t+1}$, observation $o^{t+1}_i$
    \If{$count\_based$}
        \State \textit{count} $\gets$ Count occurrences of $o_{t+1}^i$ in $b^i$
        \State $\textit{reward} \gets \frac{1}{\textit{count}}$ 
    \ElsIf{$knn\_approximation$}
        \State $D_{i, k, n} \gets$ Compute the distance of $o_{t+1}^i$ to its $k^{th}$ neighbor observation in $b^i$ for 
        \State $\textit{reward} \gets log{(D_{i, k, n} + 1)}$
    \EndIf
    \State Update $b^i$ with $ o_{t+1}^i$
    \If{${\beta}(o_{t+1}^i)$}
        \State $\textit{reward} \gets {\beta}(o_{t+1}^i)\textit{reward}$
    \EndIf
    \If{$mix\_with\_extrinsic\_reward$}
        \State $\textit{reward} \gets r_{ex}(o_{t+1}^i) + \psi \textit{reward}$
    \EndIf
    \State Yield $\textit{reward}$
\end{algorithmic}
\end{algorithm}

\subsubsection{E2C-MAPPO} 
After introducing all the components required to design E2C-MAPPO in the previous sections, Algorithm \ref{alg:e2c_mappo} shows a general template for our method. For simplicity, we show the procedure using team constraints, but the individual constraint case follows by replacing the team components with the individual ones as previously discussed. In detail, after defining the desired constraints thresholds and initializing agents' policies, value functions, and multipliers (lines 1-2), we follow the training loop of MAPPO algorithm (highlighted in italics), where agents interact in the environment to collect training data (lines 3-7). After each episode, for each agent, we perform the following steps:
\begin{itemize}
    \item Compute the OEM reward (line 10), following Algorithm \ref{alg:entropyShape}.
    \item Update the Lagrangian multipliers as described in Section \ref{sec:e2c} (line 11).
    \item Compute the (return) centralized advantage and the cost-advantages,  using the advantage functions parameterized by $\phi, \{\phi_{c_j}\}_{j = 1, \dots, m}$ (line 12).
\end{itemize}
We then update the agents' policy parameters $\{\theta_i\}_{i \in \mathcal{N}}$, the cost value functions $\{\phi_{c_j}\}_{j = 1, \dots, m}$, and the centralized value function $\phi$.

Overall, E2C-MAPPO enhances agents with an exploration-driven reward based on OEM. Once constraints are satisfied, the optimization process ``focuses" on maximizing the task objective (\textit{i.e.,} the team reward) where agents are incentivized to resume exploration by maximizing the entropy of the observation distribution.

\begin{algorithm}[t]
\caption{Template for E2C-MAPPO with team constraints}
\label{alg:e2c_mappo}
\begin{algorithmic}[1]
\State \textbf{Given:} Cost functions $c_j$ with thresholds $l_j,~\forall j = 1, \dots, m$
\State \textbf{Initialize:} 
\begin{itemize}
    \item Actors (policies) parameters $\theta_i~\forall i \in \mathcal{N}$
    \item Joint critic parameters $\phi$
    \item For each team constraint, \textit{i.e.,} $\forall j = 1, \dots, m$: 
        \begin{itemize}
            \item Joint cost-value parameters $\phi_{c_j}$
            \item Multipliers $\lambda_j = 0$
        \end{itemize}  
\end{itemize}

\For{each episode}
    \State \textit{Reset} the environment
    \For{step $t = 0, 1, \dots$}
        \State \textit{Sample} individual actions $u_{t}^i \sim \pi_i$ 
        \State \textit{Execute} the joint action $\bm{u}_t = \{u_t^i\}_{i\in\mathcal{N}}$; \textit{get} joint reward and costs and \textit{update} agents' information
    \EndFor
    \For{each agent $i$}
        \State \textbf{Model the OEM rewards} (Algorithm \ref{alg:entropyShape})
        \State \textbf{Update the Lagrangian Multipliers} (Section \ref{sec:e2c})
        \State \textbf{Compute advantage estimates} $A_\phi, \{A_{\phi_{c_j}}\}_{j = 1, \dots, m}$
        \State \textbf{Update policy's} $\theta_i$ (Equation \ref{eq:constrained_mappo})
        \State \textbf{Update Cost value functions $\phi_{c_j}$} using cost values by standard regression on mean-squared error ($\forall j = 1, \dots, m$)
    \EndFor
    \State \textbf{Update Centralized value function $\phi$} using reward values by standard regression on mean-squared error
\EndFor
\end{algorithmic}
\end{algorithm}
\section{Experiments}
\label{sec:experiments}

Our experiments aim to answer the following questions: (i) \textit{How well does E2C-MAPPO solve standard cooperative tasks compared to the unconstrained (unsafe) and constrained (safe) baseline}? (ii) \textit{How do different definitions of constraints (i.e., individual and team) affect performance}? (iii) \textit{Is policy entropy detrimental to constrained MARL performance when compared to our observation entropy?}

\noindent To answer these questions, we first evaluate E2C-MAPPO, the unconstrained MAPPO, and two constrained MAPPO baselines with and without policy entropy in increasingly complex variations of the well-known multi-rover domain \cite{agogino2004unifying}.
We then compare our framework and the two constrained baselines across three particle environment tasks~\cite{penv_original} and two safe MaMuJoCo locomotion scenarios~\cite{safe_mamujoco}, outlined below and depicted in Figure \ref{fig:environments}.\footnote{Other multiagent benchmarks such as SMAC \citep{smac} are not standard in safe MARL literature as they do not represent realistic scenarios, which makes it challenging to come up with relevant safety criteria.} Overall, in the navigation-based scenarios (\textit{i.e.,} multi-rover and particle environment), the safety requirement is \textit{collision avoidance}, while in locomotion tasks is a \textit{velocity limit}. Hence, when agents collide or exceed the maximum velocity, they receive a positive cost value and they try to limit its accumulation under defined thresholds.

% Crucially, previous related works employed environments that are either low-dimensional (e.g., grid-world domains \cite{constrained_mappo3}), or do not consider different types of environments in their evaluation \cite{constrained_mappo, constrained_mappo2}. In contrast, our experiments are performed in a variety of high-dimensional multiagent tasks that have been previously employed in MARL literature, where tight cooperation is essential~\cite{majumdar2020evolutionary, maddpg, safe_mamujoco}.

\begin{figure}[t]
    \centering
    \includegraphics[width=\linewidth]{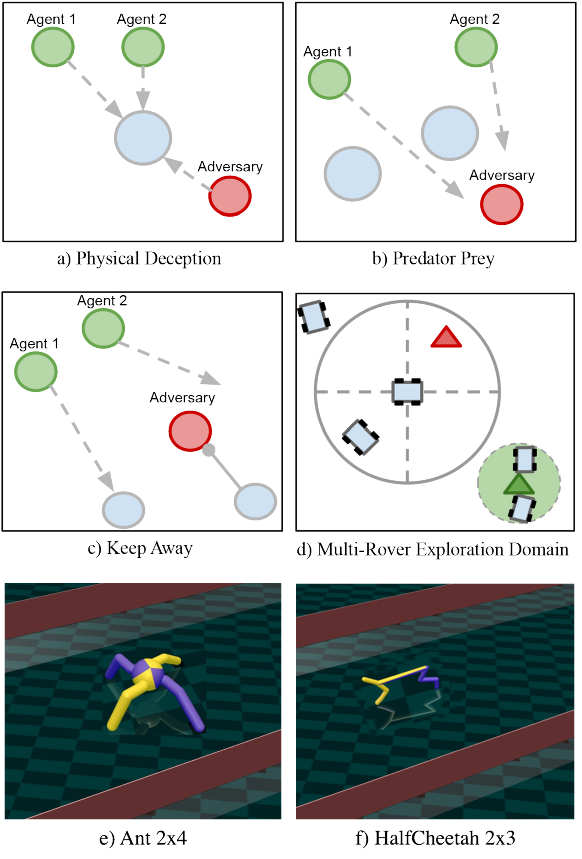}
    \caption{Particle environments (a, b, c)~\cite{penv_original}; cooperative agents are denoted by green color, adversaries by red, and landmarks by blue. Multi-rover domain (d) \citep{agogino2004unifying}; one of the agents is represented via its sensory coverage, a POI observed (green) by two agents (coupling factor is 2), and an unobserved POI (red). Multiagent Ant and HalfCheetah (e, f); each agent controls separate parts of the robot (image credit: \cite{safe_mamujoco}).}
    \label{fig:environments}
\end{figure}

\subsection{Environments}
In this section, we briefly introduce the environments considered in our evaluation. We refer to the original works for more details regarding observation and action spaces, and rewards \cite{agogino2004unifying, penv_original, safe_mamujoco}.
 
\subsubsection{Multi-Rover Exploration}
\label{sec:rover_domain}
This continuous and sparse reward (\textit{i.e.,} the team reward is only given at the end of an episode) domain (Figure \ref{fig:environments}d)~\cite{agogino2004unifying} consists of multiple agents (\textit{rovers}), and points of interest (POIs). Each rover must learn cooperative navigation skills to observe a POI simultaneously in the environment. The team size required to observe a POI is determined by a coupling factor. A higher coupling translates into a more complex coordination problem. The cost functions model collisions (\textit{i.e.,} agents receive a positive cost upon each collision) to incentivize the agents to learn collision-free navigation behaviors. These joint behaviors require taking long sequences of joint actions, which are often difficult to accomplish when incorporating constraints (Section \ref{sec:preliminaries}).  

\subsubsection{Particle Environments}
We consider three particle environments \cite{penv_original, maddpg} (Figures \ref{fig:environments}: a, b, c). For all these tasks, we consider 3 cooperative (\textit{good}) agents and 1 adversary. The latter learns a policy using the unconstrained PPO~\cite{ppo} algorithm, updating the same type of policy network as the cooperative agents. The episodic team reward totals the cooperative agents' rewards throughout an episode, and the cost functions model collisions to achieve collision-free team behaviors. In more detail, we consider the following environments:
\begin{itemize}
    \item \textit{Physical deception}: Good agents and an adversarial agent compete to reach a single landmark. Good agents cooperate to reach the landmark and receive the negative of the closest good agents' distance and the distance of the adversarial agent's distance to the landmark as their reward.
    \item \textit{Keep away}: Good agents have to reach a landmark (randomly chosen from a set of 2 landmarks) and are rewarded with their negative distance to the landmarks. An adversarial agent tries to push away the agents from their target.
    \item \textit{Predator prey}: Landmarks are used as obstacles, and good agents (with lower speed) need to capture a faster adversarial agent. When the good agents touch the adversarial agent, the agents receive positive rewards, whereas the adversary gets a penalty.
\end{itemize}

\subsubsection{Safe MaMuJoCo} 
These tasks extend the well-known single-agent locomotion benchmark to multiagent settings. Each agent controls different parts of the robot and receives a joint reward for forward movements, incentivizing learning good locomotion behaviors. We consider two tasks \cite{safe_mamujoco} (Figures \ref{fig:environments}: d, e):

\begin{itemize}
    \item \textit{Ant 2x4}: Two agents control four joints of an ant and each has to learn how to run in a corridor. The agents receive a positive cost when an ant topples over or gets too close to the wall.
    \item \textit{HalfCheetah 2x3}: Two agents control three joints of a cheetah and each has to learn how to run in a corridor. There are moving bombs in the corridor and the agents receive a positive cost when the cheetah gets too close to the bombs.
\end{itemize}

\subsection{Implementation Details}
 Data collection is performed on Xeon E5-2650 CPU nodes with 64GB of RAM. Considering the twofold nature of E2C, we call E2C-MAPPO (T) the entropy maximizing algorithm using team constraints, and E2C-MAPPO the one with individual constraints. For a fair comparison, the threshold for each agent in the individual constraint case equals the team threshold divided by the number of agents (detailed in the following section). 

The following results show the average return smoothed over the last hundred episodes of 10 runs per method with shaded regions representing the standard error. As in previous work on individual constraints \cite{cmarl_bounds}, we incorporate a GRU \citep{bahdanau2014neural} layer into the agents' networks to address the partially observable nature of some tasks and use weight sharing to speed up the training process \citep{weightsharing}. After performing an initial grid search, we use the best-performing parameters (Table \ref{table:hyper}) for MAPPO, its constrained, and E2C versions.

\begin{table}[htbp]
\centering
\caption{Hyper-parameters used in our experiments}
\vspace{-0.2cm}
\label{table:hyper} 
\catcode`,=\active
\def,{\char`,\allowbreak}
\renewcommand\arraystretch{0.95}  % More reasonable row spacing
\begin{tabular} {p{1.5cm}<{\raggedright} p{4.0cm} p{2cm}<{\raggedright}} % Adjusted column widths slightly
  \toprule
    \textbf{Component}           & \textbf{Hyper-parameter}                  & \textbf{Setting}       \\ 
  \midrule
    MAPPO                     & Clipping Coefficient                  & 0.2                   \\
                              & Discount Factor                     & 0.9                       \\
                              & GAE $\lambda$                        & 0.95                          \\
                              & Entropy Coefficient                  & 1e-3                          \\
                              & Batch Size                           & 4096                          \\
    
    Lag. Spec.                & Learning Rate                        & 0.05                          \\
                              & Lag. Multiplier                      & 1.0                  \\

    RL Actor                  & Actor Architecture                   & [128, 128]                   \\
    Spec.                     & Critic Architecture                  & [128, 128]                    \\ 
                              & Activation Func.                     & ReLU                      \\
                              & Optimizer                            & Adam \cite{kingma2014adam}                      \\
                              
   Entropy                    & Quantization Level                   & 1 (Binary)       \\
   Rewards                    & $\beta(.)$ in Rover Domain            & POI values                       \\
                              & $\beta(.)$ in Particle Environments  & Not defined                       \\
                              & $\beta(.)$ in Safe MaMuJoCo          & Not defined                       \\
                              & $k$ in PEnvs, Safe MaMuJoCo          & 5, 10   \\
    Episode                   & in Rover Domain                      & 80         \\
    Max. Len.                 & in PEnvs, Safe MaMuJoCo              & 80, 2000                       \\
  \bottomrule
\end{tabular}
\end{table}

\subsection{Results}

This section aims at showing how well the proposed approach solves the cooperative tasks and how different definitions of constraints impact the training process. 

\subsubsection{Multi-Rover Exploration}

\begin{figure}[t]
    \centering
    \includegraphics[width=1.0\linewidth]{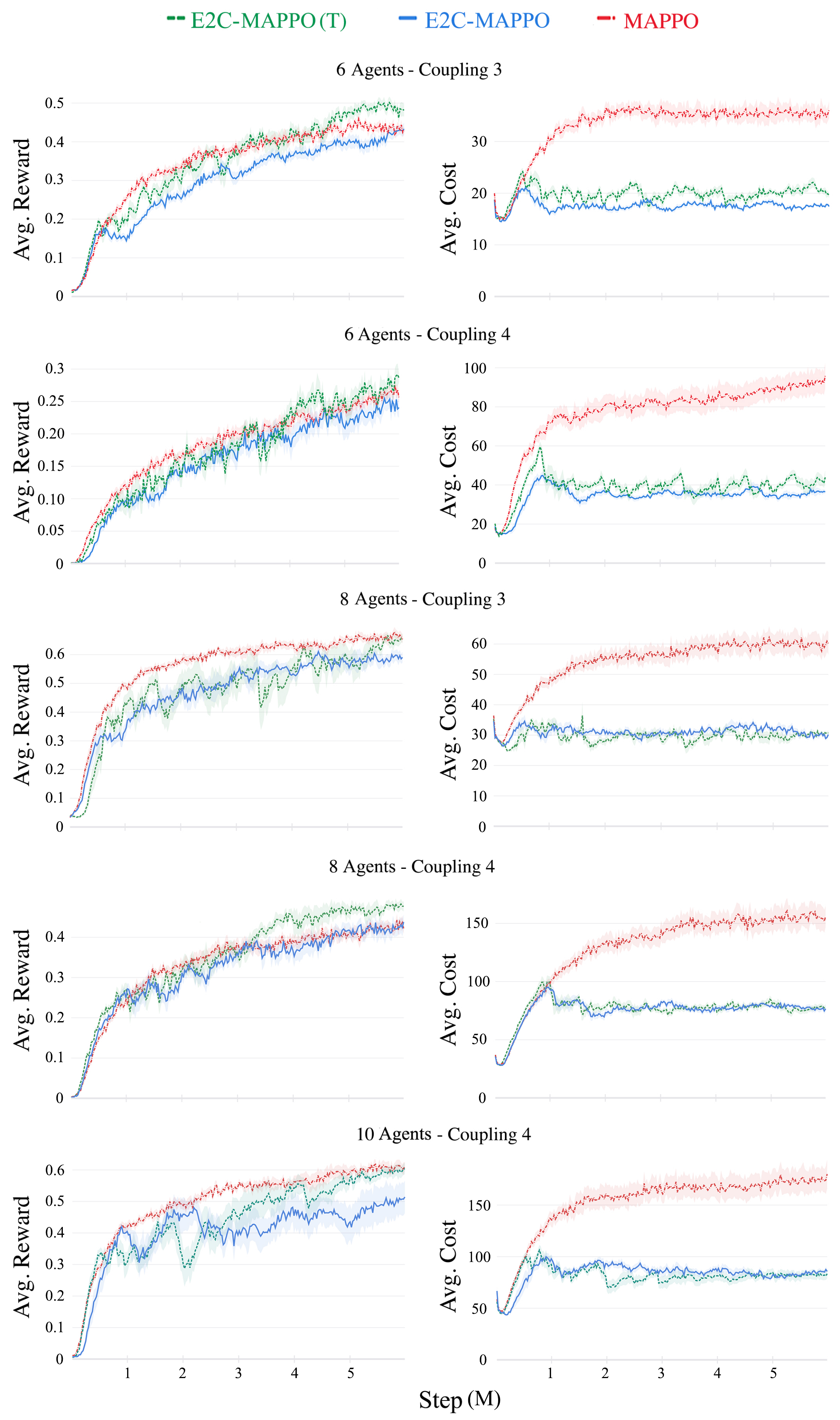}
    \caption{Average reward and cost for the unconstrained MAPPO, E2C-MAPPO and E2C-MAPPO (T) in scenarios with 6, 8, 10 rovers and coupling factors of 3, and 4. Our E2C algorithms have comparable performance to the unconstrained MAPPO while halving the unsafe actions. Moreover, the team constrained algorithm, E2C-MAPPO (T), achieves higher performance in the most complex variations of the tasks.}
    \label{fig:results}
\end{figure}

The constraint thresholds for E2C-MARL (T) is set to 20, 20, 40, 30, and 75 for the task variations listed in Figure \ref{fig:results}. Each row of the figure shows experiments with increasing number of agents and coupling. The first column highlights the average reward and the second column shows the average cost. 

In general, the E2C-MAPPO variations have comparable task-objective performance (\textit{i.e.}, average reward) to the unconstrained version, but reduce the cost by half. Crucially, when the number of rovers increases, using team constraints leads to better performance than using individual constraints. We also perform experiments in a more challenging setup with 10 rovers and a coupling of 4, which confirms the superior performance of E2C-MAPPO (T). \textit{This supports our claims regarding the potential benefits of team constraints in fully cooperative scenarios.} 

\subsubsection{Experiments with Policy Entropy}

To test our claims on the detrimental effect of policy entropy, we compare the impacts of observation entropy against policy entropy. We perform the tests using: (i) the Lagrangian MAPPO with individual constraints and policy entropy proposed by~\citet{constrained_mappo} (C-MAPPO (w. policy entropy)); and (ii) a variation not employing any entropy (C-MAPPO). For a fair comparison over these previous individual constraint baselines, we compare them against E2C-MAPPO. 
\begin{figure}[t]
    \centering
    \includegraphics[width=1.0\linewidth]{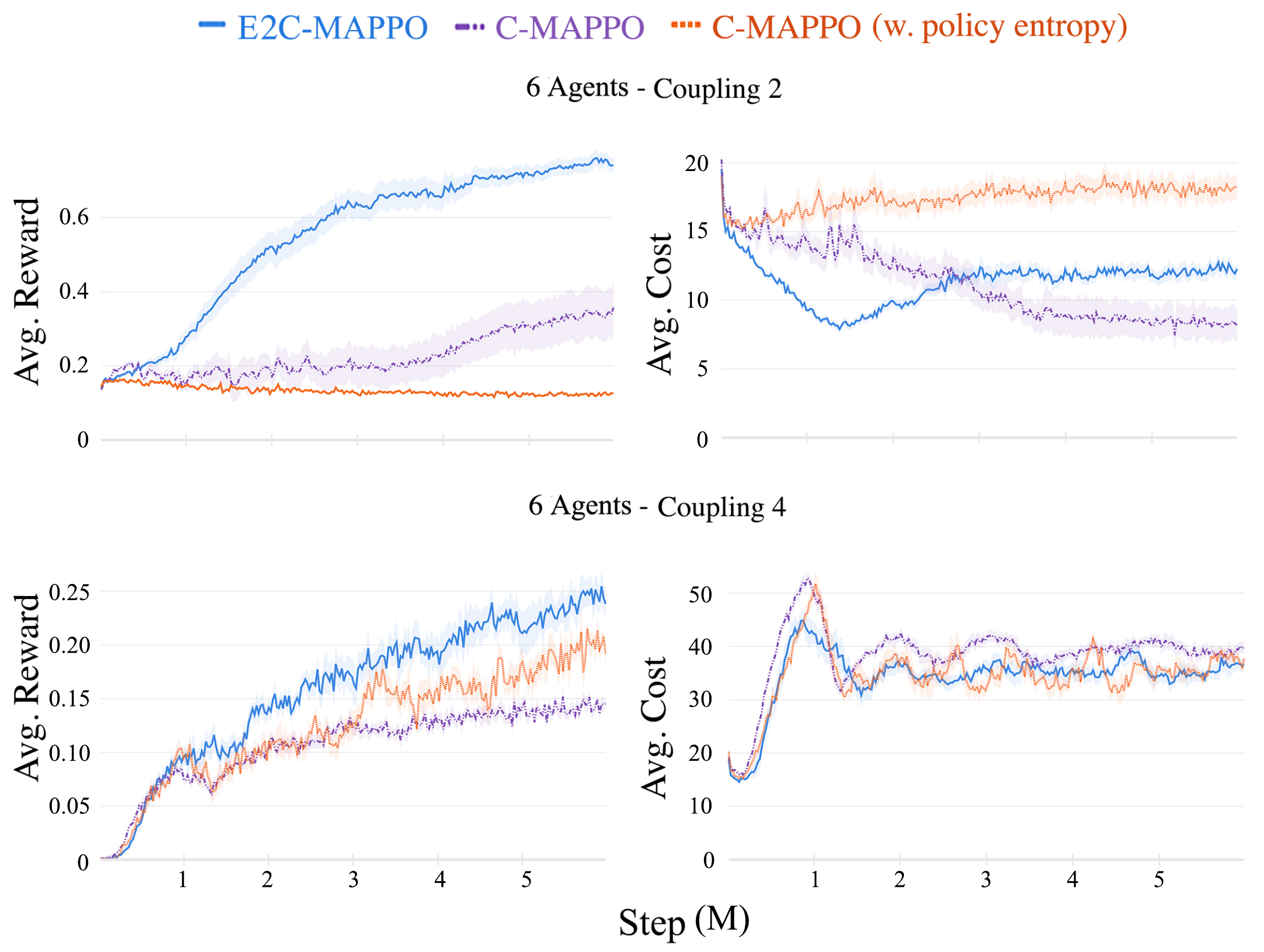}
    \caption{Impact of our observation entropy maximizing reward (E2C-MAPPO) against employing no entropy (C-MAPPO), and policy entropy as in previous constrained MAPPO algorithms (C-MAPPO (w. policy entropy)) with individual constraints \cite{constrained_mappo}. Explanatory experiments in a scenario with 6 rovers and coupling factors of 2, and 4.}
    \label{fig:ablation_results}
\end{figure}

\begin{figure*}
    \centering
    \includegraphics[width=0.8\linewidth]{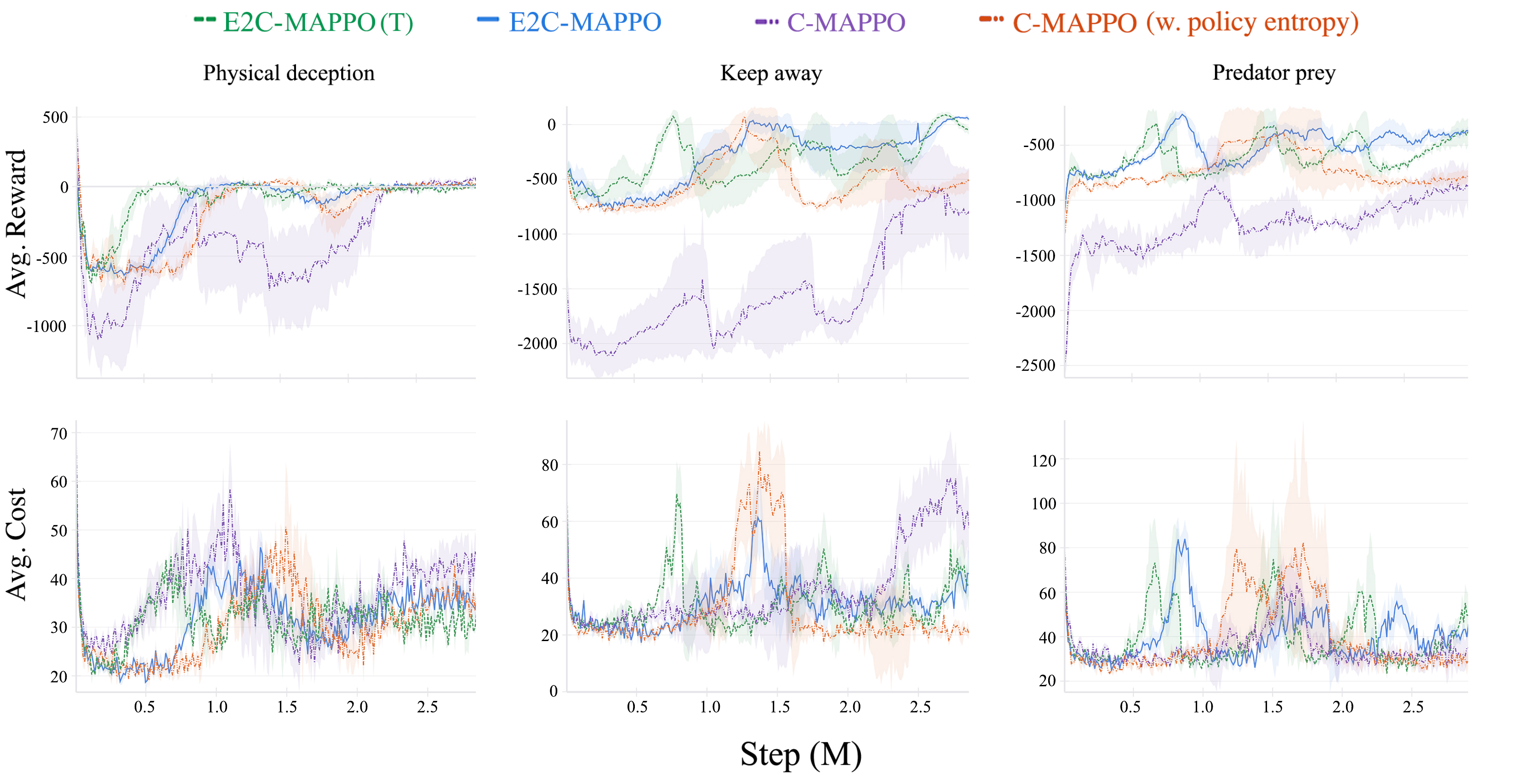}
    \caption{Average reward and cost in the particle environments.}
    \label{fig:results_mpe}
\end{figure*}

As claimed in Section \ref{sec:preliminaries}, policy entropy is detrimental under restrictive constraint thresholds. When we set the coupling and the threshold to 2, C-MAPPO with policy entropy has the lowest performance (Figure \ref{fig:ablation_results} top) as the agents' multipliers grow unbounded when failing to satisfy the constraint threshold. Moreover, while C-MAPPO has slightly lower cost than E2C-MAPPO, our method significantly outperforms in the main task performance (\textit{i.e.,} C-MAPPO struggles to learn good cooperative behaviors due to the lack of exploration). This clearly confirms the benefits of keeping the exploration active via E2C agents. Considering a coupling factor of 4, the C-MAPPO baselines achieve comparable costs and exhibits some performance due to the higher constraint thresholds. However, the information carried out by diverse observations allows E2C agents to exhibit higher performance (Figure \ref{fig:ablation_results} bottom).

\subsubsection{Particle environments}
After this preliminary experiments, we use the constrained methods, E2C-MAPPO (T), E2C-MAPPO, C-MAPPO, and C-MAPPO (w. policy entropy) for the remaining evaluations. Figure \ref{fig:results_mpe} shows the results in the particle environments, where the learning adversarial agent makes the good agents' performance brittle. Overall, all the methods achieve comparable performance in physical deception. However, E2C-MAPPO (T) converges to its peak performance in approximately half of the steps required by the C-MAPPO baselines, while improving sample efficiency over E2C-MAPPO. Cost plots show that E2C-MAPPO (T) also satisfies the constraint threshold in fewer steps, motivating its higher sample efficiency in discovering safe collaborative behaviors with higher payoffs. Results in keep away and predator prey show that the performance of C-MAPPO agents is significantly more brittle than E2C agents, as its policy entropy negatively affects performance in tasks with high uncertainty. Especially, C-MAPPO's limited exploration during the early stages of training causes agents to remain still to avoid collisions, finally lead to task failure. In contrast, E2C-MAPPO algorithms learn high-reward, low-cost behaviors without significant differences between individual and team constraints.

Overall, we note that team constraints improve performance in tasks with less uncertainty and higher cooperation while matching the performance of individual constraints under higher uncertainty.

\subsubsection{Safe MaMuJoCo}
Finally, we compare the methods in two safe MaMuJoCo tasks. In this set, we only consider E2C-MAPPO (T) as different agents use the joints of the same robot; thus, our two constraint formalizations are equivalent. We set the constraints thresholds to 5 for \textit{Ant} and 20 for {HalfCheetah}.
Figure \ref{fig:results_mamujoco} shows the results in \textit{Ant 2x4} and \textit{HalfCheetah 2x3} that confirm our claims on the benefits of observation entropy maximization in constrained setups. In more detail, E2C-MAPPO (T) achieves higher returns and lower costs (\textit{i.e.,} fewer constraint violations) in both scenarios.
\begin{figure}[t]
    \centering
    \includegraphics[width=1.0\linewidth]{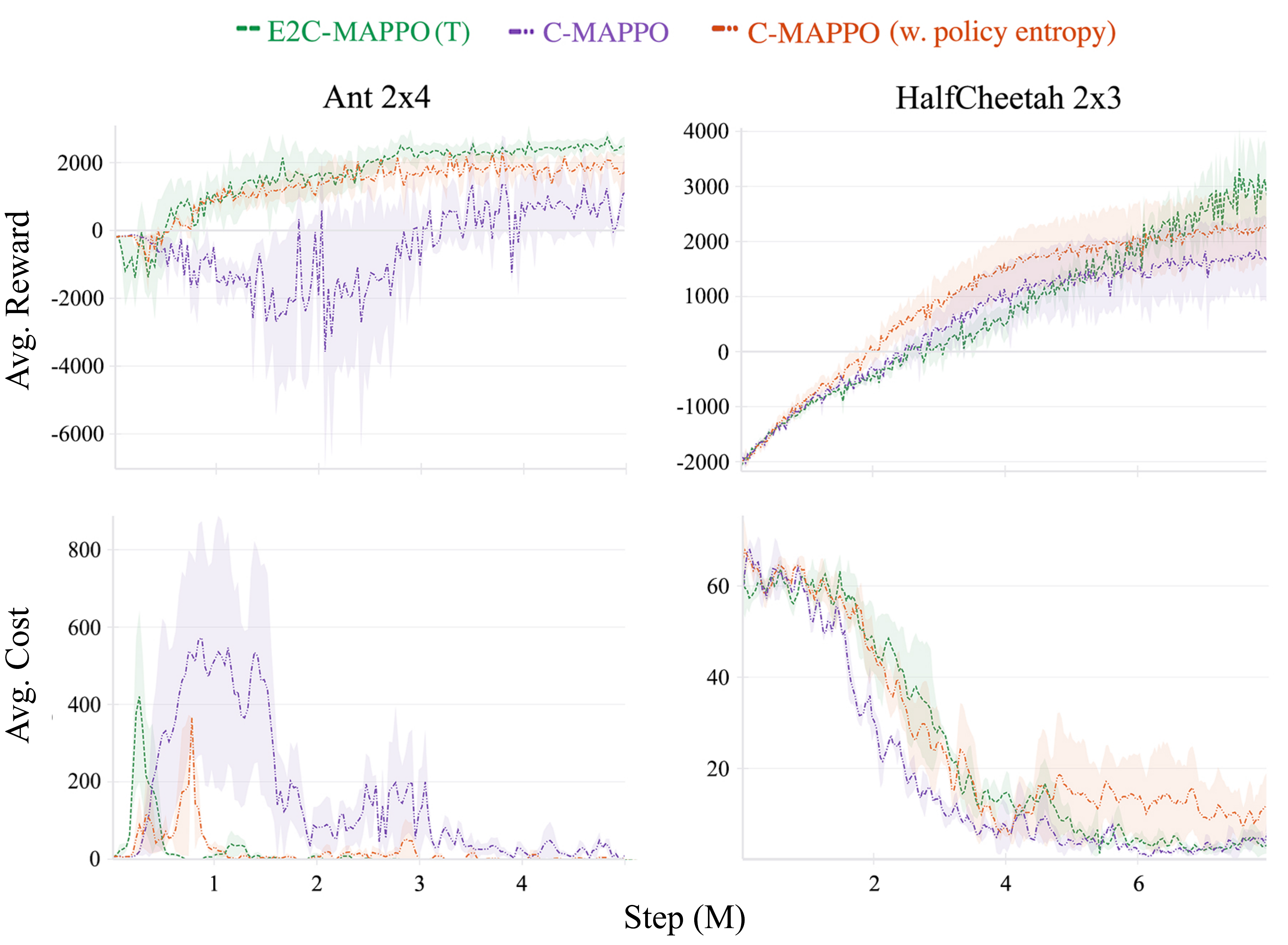}
    \caption{Average reward and cost in the safe MaMuJoCo environments. We only consider team constraints as agents represent different joints of the same robot--the individual constraint case is equivalent to the team formalization.}
    \label{fig:results_mamujoco}
\end{figure}
\section{Conclusion}
\label{sec:discussion}

We introduce entropic exploration to address the challenges of real-world applications of multiagent systems requiring cooperation and safety. We highlight the limitations of existing constrained MARL methods, which often rely on individual constraints and policy entropy maximization. These approaches can compromise task performance due to limited exploration and increased randomness in action selection.
Our approach tackles these issues by: (i) investigating team constraints, which better capture the specifics of cooperative multiagent tasks, from a theoretical and practical perspective; and (ii) leveraging observation entropy maximization (OEM) to encourage diversity in observations upon satisfying constraints. The OEM-based rewards effectively balance the competing objectives of task performance and constraint satisfaction, avoiding the pitfalls of overly conservative behavior. By prioritizing observation diversity, our E2C algorithms foster exploration even within the boundaries of strict safety constraints, enhancing both learning cooperative strategies and adherence to safety specifications.

Our experiments spanned multiple challenging environments confirms E2C's ability to consistently satisfy both individual and team constraints. Our results show the superior performance of E2C in achieving higher task performance while maintaining or improving constraint satisfaction compared to traditional baselines. These results validate the potential of OEM as a crucial mechanism to promote safe, coordinated behaviors in multiagent systems.

In conclusion, E2C offers a novel, effective solution to the exploration and safety trade-offs inherent in constrained MARL. By shifting the focus from policy randomness to observation entropy driven exploration, we provide a more principled approach to balancing performance and safety in cooperative settings. Future work could extend this framework to more complex environments, examine scalability in larger agent teams, and explore the potential of E2C in real-world applications.

%%%%%%%%%%%%%%%%%%%%%%%%%%%%%%%%%%%%%%%%%%%%%%%%%%%%%%%%%%%%%%%%%%%%%%%%

%%% The acknowledgments section is defined using the "acks" environment
%%% (rather than an unnumbered section). The use of this environment 
%%% ensures the proper identification of the section in the article 
%%% metadata as well as the consistent spelling of the heading.

% \begin{acks}
% If you wish to include any acknowledgments in your paper (e.g., to 
% people or funding agencies), please do so using the `\texttt{acks}' 
% environment. Note that the text of your acknowledgments will be omitted
% if you compile your document with the `\texttt{anonymous}' option.
% \end{acks}

%%%%%%%%%%%%%%%%%%%%%%%%%%%%%%%%%%%%%%%%%%%%%%%%%%%%%%%%%%%%%%%%%%%%%%%%

%%% The next two lines define, first, the bibliography style to be 
%%% applied, and, second, the bibliography file to be used.
\bibliographystyle{ACM-Reference-Format}
\balance
\bibliography{biblio}

%%%%%%%%%%%%%%%%%%%%%%%%%%%%%%%%%%%%%%%%%%%%%%%%%%%%%%%%%%%%%%%%%%%%%%%%

\newpage
\onecolumn
\section*{Appendix}

\section{Missing Proof in Section 3}
\label{app:proof}

\setcounter{lemma}{1}

\begin{lemma}
    Let $\pi$ and $\bar{\pi}$ be joint policies. Let $i\in\mathcal{N}$ be an agent, and $j = 1, \dots, m$ be one of the joint cost indexes. The following inequality holds:
    \begin{equation*}
        J_j(\bar{\pi}) \leq J_j(\pi) + L_j^\pi(\bar{\pi}^i)+\nu_j\sum_{h = 1}^{|\mathcal{N}|} D_{KL}^{max}(\pi^h, \bar{\pi}^h),\qquad
        \text{where}~ \nu_j = \frac{4\gamma\max_{s, \bm{u}}\vert A_j^\pi(s, \bm{u})|}{(1-\gamma)^2}.
    \end{equation*}

% \label{lemma:team_bound}
\end{lemma}

\begin{proof} The proof closely follows the one provided by \citet{constrained_mappo} for the individual constraint case. From Equations (41)-(45) of \citet{trpo} applied to joint policies $\pi, \bar{\pi}$ and the surrogate cost objective defined in terms of the joint cost advantage functions (Lemma 1), we have that:

\[
J_j(\bar{\pi}) \leq J_j(\pi) + L_j^\pi(\bar{\pi}^i)+\frac{4\gamma\alpha^2\max_{s, \bm{u}}\vert A_j^\pi(s, \bm{u})|}{(1-\gamma)^2} = J_j(\pi) + \mathbb{E}_{\bm{u}^{-i}\sim \pi^{-i}, u^i \sim \bar{\pi}^i}\left[A_j^\pi(s, \bm{u})\right] \frac{4\gamma\alpha^2\max_{s, \bm{u}}\vert A_j^\pi(s, \bm{u})|}{(1-\gamma)^2},
\]

\[
\text{where}\quad \alpha = D_{TV}^{\text{max}}(\pi, \bar{\pi}) = \max_s D_{TV}(\pi(\cdot|s), \bar{\pi}(\cdot|s)).
\]

\noindent From \citet{trpo}, we also know $D_{TV}(p, q)^2 \leq D_{KL}(p, q)$, modifying the above inequality in the following:

\[
J_j(\bar{\pi}) \leq J_j(\pi) + L_j^\pi(\bar{\pi}^i)+\frac{4\gamma\alpha^2\max_{s, \bm{u}}\vert A_j^\pi(s, \bm{u})|}{(1-\gamma)^2} = J_j(\pi) + \mathbb{E}_{\bm{u}^{-i}\sim \pi^{-i}, u^i \sim \bar{\pi}^i}\left[A_j^\pi(s, \bm{u})\right] \frac{4\gamma\max_{s, \bm{u}}\vert A_j^\pi(s, \bm{u})|}{(1-\gamma)^2}D_{KL}^{\text{max}}(\pi, \bar{\pi}),
\]

\noindent Finally, given that the following holds:

\begin{equation*}
\begin{split}
    D_{KL}^{\text{max}} &= \max_s D_{KL}(\pi(\cdot|s), \bar{\pi}(\cdot|s)) \\
    &= \max_s \left(~\sum_{h = 1}^{|\mathcal{N}|} D_{KL}\left(\pi^h(\cdot|s), \bar{\pi}^h(\cdot|s)\right)\right) \\
    &\leq \sum_{h = 1}^{|\mathcal{N}|} \max_s D_{KL}\left(\pi^h(\cdot|s), \bar{\pi}^h(\cdot|s)\right) \\
    &= \sum_{h = 1}^{|\mathcal{N}|} D_{KL}(\pi^h, \bar{\pi}^h),
\end{split}
\end{equation*}

\noindent we obtain our results:
\begin{equation*}
        J_j(\bar{\pi}) \leq J_j(\pi) + L_j^\pi(\bar{\pi}^i)+\nu_j\sum_{h = 1}^{|\mathcal{N}|} D_{KL}^{max}(\pi^h, \bar{\pi}^h),\qquad
        \text{where}~ \nu_j = \frac{4\gamma\max_{s, \bm{u}}\vert A_j^\pi(s, \bm{u})|}{(1-\gamma)^2}.
    \end{equation*}

\end{proof}
%%%%%%%%%%%%%%%%%%%%%%%%%%%%%%%%%%%%%%%%%%%%%%%%%%%%%%%%%%%%%%%%%%%%%%%%

\end{document}